\documentclass[]{article}
\usepackage{amssymb}
\usepackage{amsmath}
\usepackage{amsfonts}
\usepackage{amsthm}
\usepackage[textwidth=2cm]{todonotes}
\usepackage{graphicx}
\usepackage{subfig}
\usepackage{enumerate}
\usepackage{verbatim}
\usepackage{color}
\usepackage{float}

\usepackage{authblk}

\usepackage{tikz}
\usetikzlibrary{arrows}
\usetikzlibrary{decorations.pathmorphing,backgrounds,positioning,fit,petri}
\usepackage{wrapfig}
\usepackage{listings}
\usepackage{booktabs}
\usepackage{tabularx}

\usepackage{color}
\usepackage{float}

\usepackage{setspace}
\usepackage[left=4.5cm,top=4cm,right=4cm,bottom=3.5cm]{geometry}

\usepackage[nottoc]{tocbibind}
\newtheorem{theorem}{Theorem}[section]
\newtheorem{definition}[theorem]{Definition}
\newtheorem{lemma}[theorem]{Lemma}
\newtheorem{corollary}[theorem]{Corollary}

\theoremstyle{definition}

\newtheorem{example}[theorem]{Example}

\newcommand{\R}{\mathbb{R}} 

\newcommand{\Z}{\mathbb{Z}}
\newcommand{\M}{\mathcal{M}}
\newcommand{\I}{\mathcal{I}}
\renewcommand{\S}{\mathcal{S}}

\begin{document}
\pagestyle{plain}

\title{Uniqueness of Equilibria in Atomic Splittable Polymatroid Congestion Games}

\author[1]{Tobias Harks}
\author[2]{Veerle Timmermans \thanks{This work is part of the research programme \emph{Optimal Coordination Mechanisms for Distributed Resource Allocation} with project number 617.001.302, which is (partly) financed by the Netherlands Organisation for Scientific Research (NWO). }}

\affil[1]{Department of Mathematics, Augsburg University}
\affil[2]{Department of Quantitative Economics, Maastricht University}

\date{\today}
\maketitle

\begin{abstract}
We study uniqueness of Nash equilibria in atomic splittable congestion
games and derive a uniqueness result based on polymatroid theory: when the strategy space
of every player is a \emph{bidirectional flow polymatroid}, then equilibria are unique.
Bidirectional flow polymatroids are introduced as a subclass of polymatroids possessing 
certain exchange properties. We show that important cases such as base orderable matroids
can be recovered as a special case of bidirectional flow polymatroids.
On the other hand we show that matroidal set systems are in some sense
necessary to guarantee uniqueness of equilibria: for every atomic splittable congestion
game with at least three players and non-matroidal set systems per player, there is
an isomorphic game having multiple equilibria. Our results leave 
a gap between base orderable matroids and general matroids for which 
we do not know whether equilibria are unique.
\end{abstract}

\section{Introduction}
Congestion games as introduced in Rosenthal~\cite{Rosenthal73a} constitute an elegant  game-theoretic model describing the distributed allocation of  resources among selfish players.  
Specifically, such a game comprises a finite set of players, a finite set of resources and 
the pure strategies of a player are given by a set of allowable subsets of resources. 
In the context of \emph{network games}, the resources may correspond to edges of a graph
and the allowable subsets correspond to  paths connecting a source and a sink.
Resources have cost functions that depend on the number of players currently
using the resource. For a given strategy profile (collection of pure strategies of the players),
the disutility of each player is just the sum of resource' costs of the chosen subset of resources.
Rosenthal proved in his seminal paper that congestion games always admit a pure Nash equilibrium.

\subsection{Atomic Splittable Congestion Games}
Since the initial work of Rosenthal, several works studied related or generalized 
variants of congestion games. One such variant that we consider in this paper
are so-called \emph{atomic splittable congestion games}. In this model,  every player has a demand $d_i>0$ that she may split 
fractionally over the allowable subsets. The cost of a resource is then a function
of the total load assigned to it.
This class of games has  
applications in modeling packet-routing in communication networks (see Orda et al.~\cite{Orda93} and Korilis et al.~\cite{Korilis1997,KorilisLO95}), traffic networks (Haurie and Marcotte~\cite{Haurie85}) and logistics networks (Cominetti et al.~\cite{Cominetti09}).

Formally, there is a finite set of resources $E$, a finite set of players $N$, 
and each player $i\in N$ is associated with a weight $d_i \geq 0$ and a collection of allowable subsets of resources $\S_i \subseteq 2^E$, where $2^E$ denotes the power set of $E$. A strategy for player $i$ is then a (possibly fractional) distribution $\vec x_i\in \R_{\geq 0}^{|\S_i|}$ of the weight
over the allowable subsets $S \in $ $\S_i$. 
Thus, one can compactly represent the strategy space of every player $i\in N$ by the following polytope:
\begin{align}\label{def:strategy_space} P_i:=\{\vec{x}_i\in \R_{\geq 0}^{|\S_i|}:  \sum_{S\in\S_i} x_S= d_i\}.
\end{align}
We denote by $\vec{x}=(\vec{x}_i)_{i\in N}$ the overall strategy profile.
The induced load under $\vec{x}_i$ at $e$ is defined as $x_{i,e}:=\sum_{S\in \S_i:e\in S}x_S$
 and the total  load on $e$ is then given as $x_e:=\sum_{i\in N}x_{i,e}$. 
Resources have player-specific cost functions $c_{i,e} : \R_{\geq 0} \rightarrow \R_{\geq 0}$ which are assumed
to be non-negative, increasing, differentiable and convex.
The total cost of player $i$ in strategy distribution $\vec{x}$ is defined as
\[ \pi_{i}(\vec x)=\sum_{e\in E} c_{i,e}(x_e)\,x_{i,e}.\]
Each player wants to minimize the total cost on the used resources
and a Nash equilibrium is a strategy profile $\vec x$ from which
no player can unilaterally deviate and reduce its total cost.
Using that the strategy space is compact and cost functions are increasing and convex Kakutanis' fixed point theorem implies the existence of a Nash equilibrium \cite[Theorem 1]{Rosen65}.

\subsection{Uniqueness of Equilibria}
A fundamental property of a strategic game is the
uniqueness of equilibria. This property  is key to actually
predict the outcome of distributed resource allocation: if there are multiple equilibria
it is not clear upfront which equilibrium will be selected by the players.
This issue has been raised explicitly by Aumann~\cite{Aumann85} quoted below:
``...it is by no means clear how the players would arrive at an equilibrium, why they should play equilibrium strategies, and how a specific equilibrium would be chosen from among the set of all equilibria.''

An intriguing question in the field of atomic splittable congestion games is the possible non-uniqueness of equilibria.Let $\vec{x}$ and $\vec{y}$ be two equilibria. We say that $\vec{x}$ and $\vec{y}$ are different whenever there exists a player $i$ and resource $e$ such that $x_{i,e}\neq y_{i,e}$.
A variant on this question is whether or not there exist multiple equilibria such that there exists at least one resource $e$ for which $x_e \neq y_e$. We call this variant ``uniqueness up to induced load on the resources''. 

For non-atomic players 
and network congestion games on directed graphs, Milchtaich~\cite{Milchtaich05} proved that Nash equilibria are not unique when cost functions are player-specific. 
Uniqueness is only guaranteed if the underlying graph is two terminal $s$-$t$-\emph{nearly-parallel}. Richman and Shimkin~\cite{Richman07} extended this result to hold for atomic splittable network games. Bhaskar et al. \cite{Bhaskar15} looked at uniqueness up to induced load on the resources.
They proved that even when all players experience the same cost on a resource, there can exist multiple equilibria. They further proved that for two players, the Nash equilibrium is unique if and only if the underlying undirected graph is  generalized series-parallel. For multiple players of two types (players are of the same type if they have the same weight and share the same origin-destination pair), there is a unique equilibrium if and only if the underlying undirected graph is $s$-$t$-series-parallel. For more than two types of players, there is a unique equilibrium if and only if the underlying undirected graph is generalized nearly-parallel.

\subsection{Our Results and Outline of the Paper}
In this paper we study the uniqueness of equilibria for general set systems $(\S_i)_{i\in N}$.
Interesting combinatorial structures of the $\S_i$'s beyond paths may be
trees, forests, Steiner trees or tours all in a directed or undirected graph
or bases of matroids.

As our main result we give a sufficient condition for uniqueness based 
on the theory of polymatroids. We show that if the strategy space of every player 
is a polymatroid base polytope satisfying a special exchange property -- we term this class of polymatroids \emph{bidirectional flow polymatroids} -- the equilibria are unique.\footnote{The
formal definition of bidirectional flow polymatroids appears in Def.~\ref{def:bidirectional}.}
We demonstrate that bidirectional flow polymatroids are quite general as they
contain \emph{base-orderable matroids, gammoids, transversal and laminar matroids}. For an overview of special
cases that follow from our main result, see Figure~\ref{inclusionmatroids}. 

The uniqueness result is stated in Section~\ref{sec:main}.
In Section~\ref{sec:applications} we show that base-orderable matroids are a special case of bidirectional flow polymatroids. Definitions of polymatroid congestion games and bidirectional flow polymatroids are introduced in Sections~\ref{par:polymatroid} and~\ref{sec:bidirectional}, respectively. 

In Section~\ref{sec:nondifferentiable} and Section~\ref{sec:nonmatroids} we complement our uniqueness result by showing multiple equilibria exist when certain assumptions are dropped. In Section~\ref{sec:nondifferentiable} we discuss why it is necesarry for cost functions to be differentiable. In Section~\ref{sec:nonmatroids} we consider a game with at least three players for which the set systems $\S_i$ of all players $i\in N$ are \emph{not} bases of a matroid. Then there exists a game with strategy spaces $\phi(\S_i)$ isomorphic to $\S_i$
which admits multiple equilibria. Here, the term \emph{isomorphic} means that there is no a priori description on how the individual strategy spaces of players interweave in the ground set of resources.
Our results leave a gap between general matroids and base orderable matroids
for which we do not know whether or not equilibria are unique.

In Section~\ref{sec:undirected} we consider uniqueness of equilibria if
the set systems $\S_i$ correspond to paths in an undirected graph.
The instance used for showing multiplicity of
equilibria of non-matroid games can be seen as  a $3$-player game played on an undirected $3$-vertex cycle graph. From this we can derive a new characterization
of uniqueness of equilibria in undirected graphs.  If we assume at least three
players and if we do not specify beforehand which vertices
of the graph serve as sources or sinks, an undirected graph induces unique equilibria 
if and only if the graph has no cycle of length at least $3$.

\begin{figure}
\centering
\begin{tikzpicture}
\node[] at (0,0) {Uniform};
\node[] at (1,0) {$\subset_1$};
\node[] at (2,0) {Partition};
\node[rotate=-45] at (3.2,-0.3) {$\subset_3$};
\node[] at (4.4,-0.5) {Transversal};
\node[rotate=45] at (3.2,0.3) {$\subset_2$};
\node[] at (4.4,0.5) {Laminar};
\node[rotate=-30] at (5.7,0.3) {$\subset_4$};
\node[text width=3cm, text centered] at (7.1,-1.4) {Graphic matroid on GSP graph};
\node[rotate=45] at (5.7,-0.3) {$\subset_5$};
\node[] at (7.2,-0.6) {$\cup_6$};
\node[] at (7,0) {Gammoid};
\node[] at (8.2,0) {$\subset_7$};
\node[text width=2cm, text centered] at (9.3,0) {Strongly base orderable};
\node[] at (10.4,0) {$\subset_8$};
\node[text width=2cm, text centered] at (11.4,0) {Base orderable};
\end{tikzpicture}
\caption{Several well-known classes of matroids and the relations between them. Here GSP is short for generalized series-parallel. References and arguments for the seven inclusions can be found in Appendix~\ref{inclusionsproof}. }
\label{inclusionmatroids}
\end{figure}
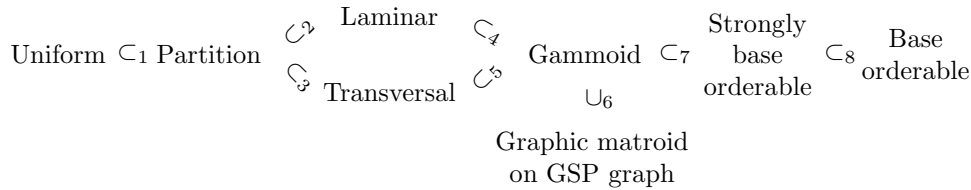

\subsection{Further Related Work}
Atomic splittable (network) congestion games have been first proposed by Orda et al.~\cite{Orda93} and Altman et al.~\cite{Altman02competitiverouting}
in the context of modeling routing in communication networks.  Other
applications include traffic and freight networks (cf. Cominetti et al.~\cite{CCS06})
and scheduling (cf. Huang~\cite{Huang11}).
Haurie and Marcotte~\cite{Haurie85} showed that classical nonatomic congestion games (cf. Beckmann et al.~\cite{Beckmann56} and Wardrop~\cite{Wardrop52})
can be modeled as atomic splittable congestion games by constructing a sequence
of games and taking the limit with respect to
the number of players. It follows that atomic splittable congestion games
are strictly more general as their nonatomic counterpart. Cominetti et al.~\cite{CCS06}, Harks~\cite{Harks:stack2011} and Roughgarden and Schoppmann~\cite{Roughgarden15} studied the price of anarchy in atomic splittable congestion games.

Matroid congestion games were first considered by Ackermann et al.~\cite{Ackermann08,Ackermann09}.
They showed that (unsplittable) weighted congestion
games possess pure Nash equilibria even for player-specific nondecreasing cost functions.
They also showed that the matroid property is the maximal property that gives
rise to a pure Nash equilibrium, that is, for any strategy space not satisfying the matroid property, there is an
instance of a weighted congestion game not having a pure Nash equilibrium. 
Integral polymatroid congestion games, a generalization of matroid congestion games, were later introduced in Harks, Klimm and Peis ~\cite{harks2014resource} (see also~\cite{HarksOosterwijkVredeveld}). In addition, polymatroid theory was recently
used in the context of nonatomic congestion games, where it is shown that matroid set systems 
are immune to the Braess paradox, see Fujishige et al.~\cite{Fujishige15}.

\section{Polymatroid Congestion Games}\label{par:polymatroid}
In polymatroid congestion games we assume that the
strategy space for every player 
corresponds to a polymatroid base polytope. 

In order to define polymatroids we first have to introduce submodular functions.
A function $\rho:2^E \rightarrow \R$ is called submodular if $\rho(U)+\rho(V) \geq \rho(U \cup V) + \rho(U \cap V)$ for all $U, V \subseteq E$.
It is called monotone if $\rho(U) \leq \rho(V)$ for all $U \subseteq V$, and normalized if $\rho(\emptyset) = 0$.
Given a submodular, monotone and normalized function $\rho$, the pair $(E,\rho)$ is called a \emph{polymatroid}. The associated \emph{polymatroid base polytope} is defined as:
\[ P_{\rho} := \left\{ \vec x \in \R_{\geq 0}^{E} \mid x(U) \leq \rho(U) 
\; \forall U \subseteq E, \; x(E) = \rho(E)\right\},\]
where $x(U) := \sum_{e \in U} x_e$ for all $U\subseteq E$.

In a polymatroid congestion game, we associate with every player $i$ a player-specific polymatroid $(E,\rho_i)$
and assume that the strategy space of player $i$ is defined by the (player-specific) polymatroid base polytope $P_{\rho_i}$.
\[ P_{\rho_i}:= \left\{ \vec x_i \in \R_{\geq 0}^{E} \mid x_i(U) \leq \rho_i(U) \; \forall U \subseteq E, \; {x_i(E)} = \rho_i(E)\right\}.\]

From now on, when we mention a  polymatroid congestion game, we mean a weighted atomic splittable polymatroid congestion game. We give three examples of polymatroid congestion games:

\begin{example}[Queueing Games (cf.~\cite{Korilis1997})]
\label{singletongames}
Let $Q=\{q_1, \dots q_m\}$ be a set of $M/M/1$ queues served in a first-come-first-served fashion and $N=\{1, \dots, n\}$  a set of companies who independently send packets with arrival rates $d_1, \dots d_n$. Every queue $q$ has a single server with exponentially distributed service time with mean $1/\mu_q$, where $\mu_q>0$. Each packet is routed to a single server $q$ out of a set of allowable queues, depending on the company. Given a distribution of packets $\vec{x} \in {\R^m_{\geq 0}}$, the mean delay of queue $q$ can be computed as $c_q(x_q)=\frac{1}{\mu_q-x_q}$. In this case the sets $\S_i$ are uniform rank-1 matroids, which are also called \emph{singleton games}.
\end{example}

\begin{example}[Transversal games]
\label{transversalgames}
{Consider a finite set $E$ of storing facilities, a finite set $A$ of locations and a finite set $N$ of players. Each player has to store an amount of $d_i$ of divisible goods in each area $j \in A$. Each area $j$ can be served from any storing facility within a given set $S_j \subseteq E$}. The sets $S_j$ may overlap, even for the same player $i$. However, due to reliability reasons, a player cannot store more than $d_j$ goods in one storing facility. The cost $c_{i,e}$ for using a specific storing facility depends on the total amount of goods that have to be stored in storing facility $e$. The more goods need to be stored, the larger the cost to use it. 

{This setting can be modeled as a bipartite graph $G$ on vertex sets $E$ and $A$, where an edge between area $j \in A$ and storage facility $e \in E$ exists if and only if area $j$ can be served from storage facility $e$. In a feasible strategy a player divides its goods over bases of the transversal matroid of this graph: subsets of storage facilities that are the endpoints of a maximal matching in $G$. Hence, the strategy space of every player $i\in N$ corresponds to the base polytope $P_{d_i \cdot rk_i}$, where $rk_i$ is the rank function of a \emph{transversal matroid}.}
\end{example}

\begin{example}[Matroid Congestion Games] \label{matroidcongestiongame}
{A \emph{matroid} $\mathcal{M}$ is a pair $(E, \mathcal{I})$, where $E$ is a finite set of resources and $\mathcal{I}$ is a family of subsets of $E$, called the \emph{independent sets}. Set $\mathcal{I}$ has the following three properties: 
\begin{enumerate}
	\item The empty set is an independent set: $\emptyset \in \mathcal{I}$.
	\item Set $\mathcal{I}$ is closed under taking subsets: if $I \subseteq J$ and $J\in \mathcal{I}$, then $I \in \mathcal{I}$.
	\item Set $\mathcal{I}$ has the \emph{exchange property}: if $I,J \in \mathcal{I}$ and $|I|<|J|$, then there exists an $e \in J$ such that $I \cup \{e\} \in \mathcal{I}$.
\end{enumerate}
A \emph{basis} is an independent set that becomes dependent on adding any element of $E$. The \emph{base set} $\mathcal{B}$ contains all bases of $(E, \mathcal{I})$.}

Consider an \emph{atomic splittable matroid congestion model}, where for every $i \in N$ the allowable subsets are the base set $\mathcal{B}_i$ of a matroid $\M_i=(E,\I_i)$. The rank function $rk_i: 2^E \rightarrow \R$ of matroid $\M_i$ is defined as: $rk_i(S):= \max \{ |U|\;|\;U \subseteq S \text{ and } U \in \I_i\} $ for all $S \subseteq E$, and is submodular, monotone and normalized  \cite{Pym70}. Moreover, the characteristic vectors of the bases in $\mathcal{B}_i$ are exactly the vertices of the polymatroid base polytope $P_{rk_i}$. It follows that the polytope $P_i:= \{\vec{x} \in {\R^{|\mathcal{B}_i|}_{\geq 0}} | \sum_{B \in \mathcal{B}_i} x_B=d_i \}$ corresponds to strategy distributions that lead to load vectors in the following polytope:
$$P_{d_i \cdot rk_i} = \left\{ \vec x_i \in {\R_{\geq 0}^{E}} | x_i(U) \leq d_i \cdot rk_i(U) \; \forall U \subseteq E, x_i(E) =d_i \cdot rk_i(E) \right\}. $$
Hence matroid congestion models are a special case of polymatroid congestion models. Both the singleton games in Example~\ref{singletongames} and the transversal games in Example~\ref{transversalgames} are a special case of matroid congestion games.
\end{example}


\section{Bidirectional Flow Polymatroids}\label{sec:bidirectional}
We provide a sufficient condition for a class of polymatroid congestion games to have a unique Nash equilibrium. We prove that if the strategy space of every player is the base polytope of a \emph{bidirectional flow polymatroid}, Nash equilibria are unique. In order to define the class of bidirectional flow polymatroids we first discuss some basic properties of polymatroids. We start with a generalization of the strong exchange property for matroids. 
Let $\chi_e\in \Z^{|E|}$ be the characteristic vector with $\chi_e(e)=1$, and $\chi_e(e')=0$ for all $e'\neq e$.
\begin{lemma}[Strong exchange property polymatroids (Murota \cite{Murota03})] \label{strongexchangeproperty}
Let $P_\rho$ be a polymatroid base polytope defined on $(E,\rho)$. Let $\vec{x},\vec{y} \in P_\rho$ and suppose $x_e > y_e$ for some $e \in E$. Then there exists an $e' \in E$ with $x_{e'} < y_{e'}$ and an $\epsilon>0$ such that:
$$\vec{x}+\epsilon(\chi_{e'}-\chi_e)\in P_\rho \text{ and } \vec{y}+\epsilon(\chi_e-\chi_{e'})\in P_\rho .$$
\end{lemma}

This exchange property will play an important role in the definition of bidirectional flow polymatroids.
Given a strategy $\vec{x}$ in the base polytope of polymatroid $(E,\rho)$, we are interested in the exchanges that can be made between $x_e$ and $x_{e'}$ for some resources in $e,e'\in E$. For that, we define a directed exchange graph $D(\vec{x})=(E,V)$, where the set of vertices equals the set of resources $E$. The edge set is $V:=\left\{{(e,e')}| \exists \; \epsilon>0 \text{ such that } \vec{x}+\epsilon({\chi_{e'}}-\chi_e)\in P_\rho \right\}$. We define exchange capacities $\hat{c}_{\vec{x}}(e,e')$ (following notation of Fujishige~\cite{Fujishige05}), which denotes the maximal amount of load that can be exchanged in $\vec{x}$ between resources $e$ and $e'$. More formally:
$$\hat{c}_{\vec{x}}(e,e'):= \max \{ \alpha |  \vec{x}+\alpha(\chi_{e'}-\chi_e)\in P_{\rho}  \}.$$

We use Lemma \ref{strongexchangeproperty} to prove the following:

\begin{lemma}  \label{symmetricdifference}Let $P_\rho$ be a polymatroid base polytope defined on $(E,\rho)$. For $\vec{x},\vec{y} \in P_\rho$, there exists a flow in $D(\vec{x})$ satisfying all supplies {and} demands, where a resource  $e$ with $x_e>y_e$ has supply of $x_e-y_e$ and $e$ with $x_e<y_e$ has a demand of $y_e-x_e$.
\end{lemma}
\begin{proof}
Consider the following algorithm:
\begin{enumerate}
\item Let $f$ be the \emph{zero flow}, a flow where we send zero flow along all edges in $D(\vec{x})$.
\item If $\vec{x}=\vec{y}$, then stop and output flow $f$.
\item Choose any element $e \in E$ such that $x_e >y_e$.
\item Use Lemma~\ref{strongexchangeproperty} to find $e' \in E$ such that $x_{e'} < y_{e'}$ and $\epsilon>0$ with  $$\vec{x}+\epsilon(\chi_{e'}-\chi_e)\in P_\rho \text{ and } \vec{y}+\epsilon(\chi_e-\chi_{e'})\in P_\rho .$$
Put  $\alpha=\min \left\{\hat{c}_{\vec{x}}(e,e'),\hat{c}_{\vec{y}}(e',e), x_{e}-y_{e}, y_{e'}-x_{e'} \right\}$, define $\vec{y} \leftarrow  \vec{y}+{\alpha}(\chi_e-\chi_{e'})$ and add $\alpha$ flow to edge $(e,e')$ in flow $f$.
\item If $\alpha < x_e-y_e$, then go to step 4. Otherwise ($\alpha = x_e-y_e$), go to step 2.
\end{enumerate}

{Note that this algorithm is a slightly changed version of Fujishige~\cite[Theorem~3.27]{Fujishige05}. The only difference is that we do not change $y$ to $x$ with exchanges that only can be made on strategy $\vec{y}$ (which is proven in  Fujishige~\cite[Theorem~3.27]{Fujishige05}) but with exchanges that can be executed on both $\vec{x}$ and $\vec{y}$. As these exchanges always exists, the results by Fujishige~\cite[Theorem~3.27]{Fujishige05} are still valid for our algorithm. Hence,} this algorithm transforms $\vec{y}$ into $\vec{x}$ with at most $\lfloor |E|^2/4 \rfloor$ elementary transformations described in Lemma~\ref{strongexchangeproperty}, such that each component $y_e$ with $y_e<x_e$ monotonically increases and each component $y_e$ with  $y_e>x_e$ monotonically decreases. Therefore $f$ satisfies all supplies and demands as described in the lemma. Flow $f$ also satisfies all capacity constraints, as every pair of resources $(e,e')$ is considered at most once, and all exchanges can be done on $\vec{x}$. Hence $f_{(e,e')} \leq \hat{c}_{\vec{x}}(e,e')$, thus $f$ is a flow in $D(\vec{x})$ satisfying all supplies and demands.
\end{proof}

The flow $f$ mentioned in Lemma \ref{symmetricdifference} is a flow from the perspective of strategy $\vec{x}$ and therefore we call this a \emph{directed flow}. In the following we define a \emph{bidirectional flow}. Let $P_\rho$ again be a polymatroid base polytope on set $E$. For any $\vec{x},\vec{y} \in P_\rho$ define the capacitated graph $D(\vec{x},\vec{y})$ on vertices $E$. An edge $(e,e')$ exist if there is an $\epsilon>0$ such that $\vec{x}+\epsilon(\chi_e'-\chi_e)\in P_\rho$ and $\vec{y}+\epsilon(\chi_e-\chi_e')\in P_\rho$. For edges $(e,e')$ we define capacities $\hat{c}_{\vec{x},\vec{y}}(e,e')$ as follows:
$$\hat{c}_{\vec{x},\vec{y}}(e,e'):= \max \{ \alpha |  \vec{x}+\alpha(\chi_{e'}-\chi_e)\in P_\rho \text{ and } \vec{y}+\alpha(\chi_{e}-\chi_e')\in P_\rho  \}$$
A \emph{bidirectional flow} is a flow in $D(\vec{x},\vec{y})$ where every resource  $e$ with $x_e>y_e $ has supply of $x_{e}-y_{e}$ and every resource $e$ with $x_e<y_e$ has a demand of $y_{e}-x_{e}$. Such a flow might not {exist}. In that case we say that $\vec{x}$ and $\vec{y}$ are \emph{conflicting strategies}.

We are ready to define the class of \emph{bidirectional flow polymatroids}:

\begin{definition}[Bidirectional flow polymatroid]\label{def:bidirectional}
A polymatroid $(E,\rho)$ is called a \emph{bidirectional flow polymatroid} if for every pair of vectors $\vec{x},\vec{y}$ in {the} base polytope $P_\rho$,  there exists a bidirectional flow in $D(\vec{x},\vec{y})$.
\end{definition} 
We give a simple example of a bidirectional flow polymatroid.

{	
	\begin{example}
	We consider polymatroid $P_{\rho}$ defined by the graphic matroid on the graph depicted in Figure~\ref{flowexample}. In this polymatroid, a total load of $1$ is divided over the bases of the graphic matroid. Here, for any two strategies $\vec x$ and $\vec y$ there exists a bidirectional flow in $D(\vec{x},\vec{y})$. In particular, in Figure~\ref{flowexample2} we show the existance for a bidirectional flow for strategy $\vec x$ and $\vec y$ defined in Figure~\ref{flowexample}.
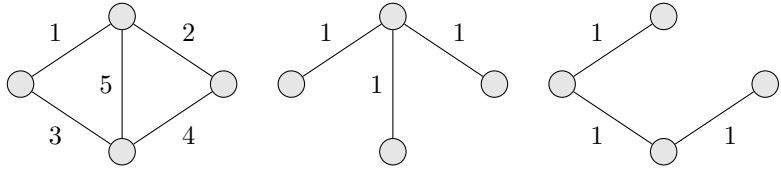
\begin{figure}[h!]
	\centering
	\begin{tikzpicture}[scale=0.45,main_node/.style={circle,fill=black!10,draw,minimum size=1em,inner sep=1pt}]
	\tikzset{edge/.style ={thick}}
	\tikzset{cut/.style ={thick,fill=red!100,bend right}}
	\tikzset{invis/.style ={fill=white!10}}
	\node[main_node] (1) at (0,0) {};
	\node[main_node] (2) at (3,2) {};
	\node[main_node] (3) at (3,-2) {};
	\node[main_node] (4) at (6,0) {};
	\path[] (1) edge [] node [above left] {1} (2);
	\path[] (2) edge [] node [above right] {2} (4);
	\path[] (1) edge [] node [below left] {3} (3);
	\path[] (3) edge [] node [below right] {4} (4);
	\path[] (2) edge [] node [left] {5} (3);	
	\node[main_node] (5) at (8,0) {};
	\node[main_node] (6) at (11,2) {};
	\node[main_node] (7) at (11,-2) {};
	\node[main_node] (8) at (14,0) {};
	\path[] (5) edge [] node [above left] {1} (6);
	\path[] (6) edge [] node [above right] {1} (8);
	\path[] (6) edge [] node [left] {1} (7);	
	\node[main_node] (9) at (16,0) {};
	\node[main_node] (10) at (19,2) {};
	\node[main_node] (11) at (19,-2) {};
	\node[main_node] (12) at (22,0) {};	
	\path[] (10) edge [] node [above left] {1} (9);
	\path[] (11) edge [] node [below left] {1} (9);
	\path[] (12) edge [] node [below right] {1} (11);
	\end{tikzpicture}
	\caption{Left: the original graph with numbered resources. Middle: Load distribution for strategy $\vec x$. Right: Load distribution for strategy $\vec y$.}
	\label{flowexample}
\end{figure} 
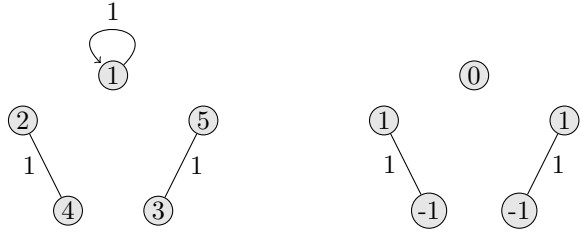
\begin{figure}[h!]
	\centering
	\begin{tikzpicture}[scale=0.6,main_node/.style={circle,fill=black!10,draw,minimum size=1em,inner sep=1pt}]
	\tikzset{edge/.style ={thick}}
	\tikzset{cut/.style ={thick,fill=red!100,bend right}}
	\tikzset{invis/.style ={fill=white!10}}
	\node[main_node] (4) at (0,0) {4};
	\node[main_node] (3) at (2,0) {3};
	\node[main_node] (5) at (3,2) {5};
	\node[main_node] (1) at (1,3) {1};
	\node[main_node] (2) at (-1,2) {2};
	\path[loop] (1) edge [] node [above] {1} (1);
	\path[] (3) edge [] node [right] {1} (5);
	\path[] (2) edge [] node [left] {1} (4);
	\node[main_node] (9) at (8,0) {-1};
	\node[main_node] (8) at (10,0) {-1};
	\node[main_node] (10) at (11,2) {1};
	\node[main_node] (6) at (9,3) {0};
	\node[main_node] (7) at (7,2) {1};
	\path[] (8) edge [] node [right] {1} (10);
	\path[] (7) edge [] node [left] {1} (9);
	\end{tikzpicture}
	\caption{Left: Graph $D(\vec{x},\vec{y})$ with corresponding capacities.  Right: the bidirectional flow in $D(\vec{x},\vec{y})$, including supplies and demands.}
	\label{flowexample2}
\end{figure} 
\end{example}
}

\section{A Uniqueness Result}\label{sec:main}
In this section we prove that when the strategy space of every player is the base polytope of a bidirectional flow polymatroid, equilibria are unique. We denote the \emph{marginal cost} of player $i$ on  resource $e\in E$ by $\mu_{i,e}(\vec{x}) = c_{i,e}(x_e)+x_{i,e} c_{i,e}'(x_e).$

An equilibrium condition for polymatroid congestion games, a result that follows from~\cite[Lemma 1]{Harks:stack2011}, is as follows:
\begin{lemma}
\label{equalresource}
Let  $\vec{x}$ be a Nash equilibrium in a polymatroid congestion game.  If $x_{i,e}>0$, then for all $e' \in E$ for which  there is an $\epsilon>0$ such that
$\vec{x_i}+\epsilon(\chi_{e'}-\chi_e) \in P_{\rho_i}$, we have $\mu_{i,e}(\vec{x}) \leq \mu_{i,e'}(\vec{x})$.
\end{lemma}

In the rest of this section we will prove the following theorem:

\begin{theorem} \label{uniquenessresult}
If for a polymatroid congestion game, the strategy space for every player is the base polytope of a bidirectional flow polymatroid, then the equilibria of this game are unique.
\end{theorem}

From now on we assume $\vec{x}=(\vec{x}_i)_{i\in N}$ and $\vec{y}=(\vec{y}_i)_{i \in N}$ are strategy profiles, where strategies $\vec{x}_i$ and $\vec{y}_i$ are taken from the  base polytope  $P_{\rho_i}$ of a player-specific bidirectional flow polymatroid. Before we  prove Theorem~\ref{uniquenessresult}, we first  introduce some new notation. We define $E^+ = \{e \in E| x_e > y_e \}$ and $E^- = \{e \in E| x_e < y_e \}$ as the sets of \emph{globally} overloaded and underloaded resources. Define $E^==\{e \in E| x_e=y_e \}$ as the set of resources on which the total load does not change.  In the same way we define player-specific sets of \emph{locally} underloaded and overloaded resources $E^{i,+} = \{e \in E| x_{i,e} > y_{i,e} \}$ and $E^{i,-} = \{e \in E| x_{i,e} < y_{i,e} \}$. We also introduce four player sets:
\begin{eqnarray*}
N^+_> = \{i \in N| \sum_{e \in E^+} x_{i,e} - y_{i,e} > 0 \}, \; \; \; \; \; \;    N^-_>  =  \{i \in N| \sum_{e \in E^- {\cup E^=}} x_{i,e} - y_{i,e} > 0 \}, 		\\
N^+_< 	= \{i \in N| \sum_{e \in E^+} x_{i,e} - y_{i,e} \leq 0 \}, \; \; \; \; \; \;	 N^-_<	= \{i \in N| \sum_{e \in E^-{\cup E^=}} x_{i,e} - y_{i,e} \leq 0 \}.\\
\end{eqnarray*}
We can distinguish between two cases. Either $E = E^=$, thus $x_e=y_e$ for all resources $e \in E$, or $E \neq E^=$, which implies that $E^+$ and $E^-$ are non-empty. 

\begin{lemma} \label{nplus}If $E \neq E^=$, then $N^+_> \neq \emptyset$. \end{lemma}
\begin{proof}
Every player distributes the same weight over the resources in $\vec{x}_i$ and $\vec{y}_i$, thus $\sum_{e \in E} x_{i,e} - y_{i,e} = 0$ and $N^+_> = N^-_<$ and $N^+_< = N^-_>$. As $E^+ \neq \emptyset$ we have:
$$
0<\sum_{e \in E^+} x_e-y_e= \sum_{i\in N^+_>} \sum_{e \in E^+}  x_{i,e}-y_{i,e} + \sum_{i\in N^+_{{<}}}\sum_{e \in E^+}  x_{i,e}-y_{i,e}.$$
Note that the first term in the last expression is non-negative and the second one is non-positive. As the whole equation should be positive, we need that this first term is strictly positive and therefore $N^+_> \neq \emptyset$.
\end{proof}

For each player $i$ we create a graph $G(\vec{x}_i,\vec{y}_i)$ from graph $D(\vec{x_i},\vec{y_i})$ by adding a super-source  $s_i$ and a super-sink $t_i$ to $D(\vec{x_i},\vec{y_i})$. We add edges from $s_i$ to $e \in E^{i,+}$ with capacity $x_{i,e}-y_{i,e}$ and edges from $e\in E^{i,-}$ to $t_i$ with capacity $y_{i,e}-x_{i,e} $. Graph $G(\vec{x}_i,\vec{y}_i) {=(V_G,E_G)}$ is visualized in Figure \ref{graphdpath}.
 
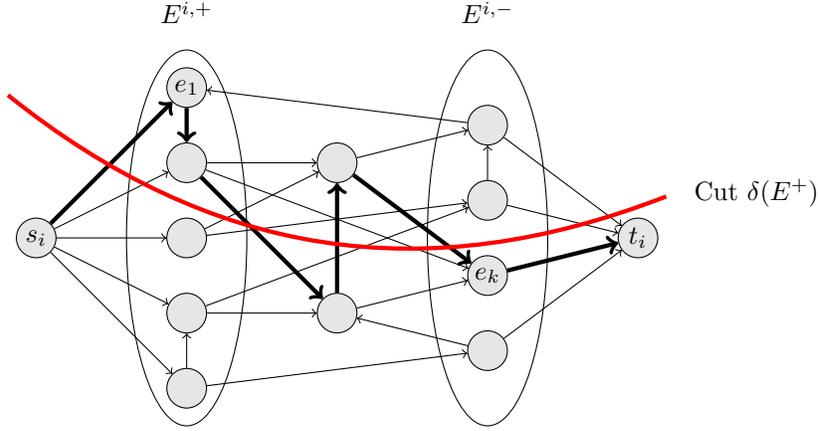
\begin{figure}
\centering
\begin{tikzpicture}[main_node/.style={circle,fill=black!10,draw,minimum size=1.5em,inner sep=1pt}]
\tikzset{edge/.style ={->}}
\tikzset{edgep/.style ={ultra thick,fill=red!100,->}}
\tikzset{edged/.style ={->,dashed}}
\tikzset{edgedp/.style ={ultra thick,->,dashed}}
\tikzset{cut/.style ={thick,fill=red!100,bend right}}
\tikzset{invis/.style ={fill=white!10}}
   \node[main_node] (1) at (0,1) {$s_i$};
   
    \node[main_node] (2) at (2, 3)  {$e_1$};
    \node[main_node] (3) at (2, 2) {};
    \node[main_node] (4) at (2, 1) {};
    \node[main_node] (5) at (2, 0) {};
    \node[main_node] (6) at (2, -1) {};

    \node[main_node] (9) at (6, 2.5)  {};
    \node[main_node] (10) at (6, 1.5) {};
    \node[main_node] (11) at (6, 0.5) {$e_k$};
    \node[main_node] (12) at (6, -0.5) {};

	\node[main_node] (13) at (4, 2) {};
	\node[main_node] (14) at (4, 0) {};

    \node[main_node] (16) at (8, 1) {$t_i$};

\node[draw=none] at (2,4) {$E^{i,+}$};
\node[draw=none] at (6,4) {$E^{i,-}$};

 \draw (2,1) ellipse (0.8cm and 2.5cm);
  \draw (6,1) ellipse (0.8cm and 2.5cm);

 \draw[edge] (9) -- (2);
 \draw[edgep] (2) -- (3);
 \draw[edgep] (3) -- (14);
 \draw[edgep] (14) -- (13);
 \draw[edgep] (13) -- (11);
 \draw[edge] (3) -- (13);
 \draw[edge] (3) -- (11);
 \draw[edge] (4) -- (10);
 \draw[edge] (4) -- (13);
 \draw[edge] (5) -- (10);
 \draw[edge] (5) -- (14);
 \draw[edge] (6) -- (5);
 \draw[edge] (6) -- (12);
 \draw[edge] (13) -- (9);
 \draw[edge] (14) -- (11);
 \draw[edge] (10) -- (9);
 \draw[edge] (12) -- (14);

    \draw[edgep] (1) -- (2);
    \draw[edge] (1) -- (3);
    \draw[edge] (1) -- (4);
    \draw[edge] (1) -- (5);
    \draw[edge] (1) -- (6);

    \draw[edge] (9) -- (16);
    \draw[edge] (10) -- (16);
    \draw[edgep] (11) -- (16);
    \draw[edge] (12) -- (16);

\node[draw=none] (lcut) at (-0.5,3) {};
\node[draw=none] (rcut) at (8.5,1.6) [label=right: Cut $\delta(E^+)$]{};
\path[] (lcut) edge [ultra thick,red,bend right] (rcut);

\end{tikzpicture}
\caption{Visualization of graph $G(\vec{x}_i,\vec{y}_i)$ and cut $\delta(E^+)$ used in the proof of Lemma \ref{graphg2path}.}
\label{graphdpath}
  
\end{figure}

Recall that strategies $\vec{x_i}$ and $\vec{y_i}$ are both chosen from the base polytope of a bidirectional flow polymatroid. Therefore there exists a flow $f_i$ in $D(\vec{x_i},\vec{y_i})$ where every resource  $e \in E^{i,+}$ has a supply of $x_{i,e}-y_{i,e}$ and $e \in E^{i,-}$ a demand of $y_{i,e}-x_{i,e}$. Using $f_i$ we define a flow $f'_i$ in $G(\vec{x}_i,\vec{y}_i)$ as follows:

\begin{equation}
\label{flowfprime}
f_i'(e,e')= 
\begin{cases}
    x_{i,e}-y_{i,e},              & \text{if } e=s_i \text{ and } e' \in E^{i,+},\\
    y_{i,e}-x_{i,e},              & \text{if } e\in E^{i,-} \text{ and } e' =t_i, \\
    f_i(e,e'),& \text{otherwise.}\\
\end{cases}
\end{equation}

\begin{lemma} \label{graphg2path}
There exists a player $i$ and a path $(s_i,e_1, \dots, e_k,t_i)$ in $G(\vec{x}_i,\vec{y}_i)$ such that $e_1 \in E^{i,+} \cap (E^+ \cup E^=)$ and  $e_k \in E^{i,-} \cap (E^- \cup E^=)$.
\end{lemma}
\begin{proof}
If $E \neq E^=$,  then using Lemma~\ref{nplus} we have that $N^+_> \neq \emptyset$, and we pick a player $i \in N^+_>$. Flow $f'_i$ can be decomposed into flow carrying $s_i$-$t_i$ paths, and we will show that there exists a path in this path decomposition that goes from $s_i$ to a vertex $e_1 \in E^{i,+} \cap E^+$, and, after visiting possibly other vertices, finally goes through a vertex $e_k \in E^{i,-} \cap (E^- \cup E^=)$ to $t_i$.
To see this consider the cut $\delta(E^+):={\{ (u,v) \in E_G \mid u \in E^+ \text{ and } v \notin E^+, \text{ or } u \notin E^+ \text{ and } v \in E^+ ) \}}$, as visualized in Figure~\ref{graphdpath}. Recall that $i\in N^+_>$, hence, $\sum_{e\in E^+}x_{i,e}-y_{i,e} >0$. Thus, in $f'_i$ more load enters $E^+$ from $s_i$, than leaves $E^+$ to $t_i$. This implies that in the flow decomposition of $f'_i$ there must be a path that goes from $s_i$ to a vertex $e_1 \in E^{i,+} \cap E^+$, crosses cut $\delta(E^+)$ an odd number of times to a vertex  $e_k \in E^{i,-} \cap (E^- \cup E^=)$ before ending in $t_i$. As this is a flow-carrying path in $f'_i$, it exists in $G(\vec{x}_i,\vec{y}_i)$. 

If $E = E^=$, pick any player $i$ for which there exists a resource $e$ with $x_{i,e}\neq y_{i,e}$ and look at the path decomposition of $f'_i$. Every path  $(s_i,e_1, \dots, e_k,t_i)$ in this decomposition is a path such that $e_1 \in E^{i,+}$ and  $e_k \in {E^{i,-}}$.
 As $E=E^=$, it also holds that $e_1 \in E^{i,+} \cap E^=$ and  $e_k \in E^{i,-} \cap E^=$. As this is a flow-carrying path in $f'_i$, it exists in $G(\vec{x}_i,\vec{y}_i)$
\end{proof}

\begin{proof}[Proof of Theorem \ref{uniquenessresult}]
Assume $\vec{x}$ and $\vec{y}$ are both Nash equilibria. Using Lemma~\ref{graphg2path} we find a path $(s_i,e_1, \dots, e_k,t_i)$ in $G(\vec{x}_i,\vec{y}_i)$ such that $e_1 \in E^{i,+} \cap (E^+ \cup E^=)$ and  $e_k \in E^{i,-} \cap (E^-  \cup E^=)$. Since every edge $(e_j,e_{j+1})$ exists in $G(\vec{x}_i,\vec{y}_i)$,  for all $j \in \{1, \dots, k-1 \}$ we get {for sufficiently small $\epsilon >0$}:
$$\vec{x_i}+\epsilon(\chi_{e_{j+1}}-\chi_{e_j})\in P_{\rho_i} \text{ and } \vec{y_i}+\epsilon(\chi_{e_j}-\chi_{e_{j+1}})\in P_{\rho_i}.$$

Using Lemma~\ref{equalresource} we obtain for $\vec{x}$:
\begin{equation}\label{equation1} \mu_{i,e_1}(\vec{x}) \leq \mu_{i,e_2}(\vec{x}) \leq \dots \leq  \mu_{i,e_k}(\vec{x}), \end{equation}
and similarly for $\vec{y}$:
\begin{equation}\label{equation2}\mu_{i,e_k}(\vec{y}) \leq \mu_{i,e_{k-1}}(\vec{y}) \leq \dots \leq \mu_{i,e_1}(\vec{y}). \end{equation}

Recall that $\mu_{i,e}(\vec{x}) = c_{i,e}(x_e)+x_{i,e} c_{i,e}'(x_e)$. As $e_1 \in E^{i,+}$, we have that  $x_{i,e_1} > y_{i,e_1}$. Because $c_{i,e_1}$ is strictly increasing and $e_1 \in (E^+  \cup E^=)$ we get $c_{i,e_1}(x_{e_1})\geq  c_{i,e_1}(y_{e_1})$ and $c'_{i,e_1}(x_{e_1})>0$ using $x_{e_1}\geq x_{i,e_1}>0$. Moreover, since $c_{i,e_1}$ is convex, the slope of $c_{i,e_1}$ is non-decreasing and, hence, $c'_{i,e_1}(x_{e_1})\geq  c'_{i,e_1}(y_{e_1})$.
Putting things together, we get
\begin{equation}\label{equation3}\mu_{i,e_1}(\vec{y}) < \mu_{i,e_1}(\vec{x}). \end{equation}

Similarly, as $e_k \in E^{i,-} \cap (E^-  \cup E^=)$, we have:  
\begin{equation}\label{equation4} \mu_{i,e_k}(\vec{x}) \leq \mu_{i,e_k}(\vec{y}). \end{equation}

Combining~\eqref{equation1},~\eqref{equation2},~\eqref{equation3}  and~\eqref{equation4}, we have: 
\begin{equation*} \label{costcontradiction} \mu_{i,e_k}(\vec{x}) \leq \mu_{i,e_k}(\vec{y}) \leq \mu_{i,e_1}(\vec{y}) < \mu_{i,e_1}(\vec{x}) \leq \mu_{i,e_k}(\vec{x}).  \end{equation*}

This is a contradiction and therefore either strategy $\vec{x_i}$ or $\vec{y_i}$ is not a Nash equilibrium for player $i$. 
\end{proof}

\section{Applications}\label{sec:special}
\label{sec:applications}

In this section we demonstrate that bidirectional flow polymatroids are general enough to allow for meaningful applications. As described in Example~\ref{matroidcongestiongame}, matroid congestion games belong to polymatroid congestion games. A subclass of matroids are \emph{base orderable} matroids introduced by Brualdi~\cite{Brualdi69} and Brualdi and Scrimger~\cite{Brualdi68}. 

\begin{definition}[Base orderable matroid]
A matroid $\mathcal{M}=(E,\mathcal{I})$ is called base orderable if for every pair of bases $(B,B')$ there exists a bijective function $g_{B,B'}:B \rightarrow B'$ such that both  $B-e+g_{B,B'}(e) \in \mathcal{B}$ and  $B'+e-g_{B,B'}(e) \in \mathcal{B}$ { for all $e \in E$}.
\end{definition}
We prove that polymatroids defined by the rank function of a base orderable matroid belong to the class of bidirectional flow polymatroids. Therefore, all matroid congestion games for which the player-specific matroids are base orderable have unique equilibria.
\begin{theorem} \label{baseorderableproof} 
Let $rk$ be the rank function of a base orderable matroid $M=(E, rk)$.
Then, for any $d\geq 0$, the polymatroid $(E, d \cdot rk)$  is a bidirectional flow polymatroid.
\end{theorem}
\begin{proof}
{Similar as in Example~\ref{matroidcongestiongame},} polytope $P$ describes how weight $d$ can be divided over bases in $\mathcal{B}$ to obtain a feasible strategy $\vec{x} \in  P_{d \cdot rk}$. We call vector $\vec{x}' \in P$ a \emph{base decomposition} of $\vec{x}$ if it satisfies $x_{e} = \sum_{B \in \mathcal{B}; e\in B} x'_B$ for all $e\in E$. 
{Note that a base composition of $\vec{x} \in P_{d \cdot rk}$ always exists, as $P_{d \cdot rk}$ is the convex hull of all characteristic vectors (multiplied by $d$) of all the bases of matroid $M$ (see~\cite[Corollary 3.25]{Fujishige05}).}  
Given two vectors $\vec{x},\vec{y} \in P_{d \cdot rk}$, we look at the differences between two base decompositions $\vec{x}', \vec{y}' \in  P$. We introduce sets $\mathcal{B}^+ ,\mathcal{B}^- \subset \mathcal{B}$ that will contain respectively the \emph{overloaded} and \emph{underloaded bases}: $\mathcal{B}^{+} = \{B \in \mathcal{B} | x'_B > y'_B \}$ and $\mathcal{B}^{-} = \{B \in \mathcal{B}| x'_B < y'_B \}$.

Using these sets we create the complete directed bipartite graph $D_{\mathcal{B}}(\vec{x}, \vec{y})$ on vertices $(\mathcal{B}^{+},\mathcal{B}^{-})$, where bases $B \in  \mathcal{B}^{+}$ have a supply $ x'_{B} - y'_{B} $ and bases $B \in  \mathcal{B}^{-}$ have a demand $ y'_{B} - x'_{B}$. As the total supply equals the total demand, there exists a transshipment $t$ from strategies  $B\in \mathcal{B}^{+}$ to strategies $B' \in \mathcal{B}^{-}$, such that, when carried out, we obtain $\vec{y}'$ from $\vec{x}'$. 
We denote by $t_{(B,B')}$
the amount of load transshipped from $B\in \mathcal{B}^{+}$ to $B' \in \mathcal{B}^{-}$. 

In the remainder of the proof, we use transshipment $t$ to construct a flow $f$ in graph $D(\vec{x},\vec{y})$. As the polymatroid is defined by the rank function of a base orderable matroid, for every pair of bases $(B,B')$ there exists a bijective function $g_{B,B'}:B \rightarrow B'$ such that both  $B-e+g_{B,B'}(e) \in \mathcal{B}$ and  $B'+e-g_{B,B'}(e) \in \mathcal{B}$ for all $e \in B$. Note that when $e\in B \cap B'$, $g_{B,B'}(e)=e$. {Using the function $g_{B,B'}$, we can decompose the value transshipped from $B$ to $B'$ into a transshipment between resources. For all combinations of resources $(e,e') \in E \times E$ we define:}
$$\mathcal{B}^2_{e,e'}:=\left\{ (B,B')\in \mathcal{B}^+ \times \mathcal{B}^- | e \in B, e'\in B' \text{ and } g_{B,B}(e)=e' \right\}.$$ 
We define flow $f$ as: $f_{(e,e')}=\sum_{(B,B') \in \mathcal{B}^2_{e,e'} } t_{B,B'}$ for all $(e,e') \in E \times E$. {Then $f$ has the following two properties: 
\begin{enumerate}
	\item It satisfies all demands and supplies in $D(\vec{x},\vec{y})$ as $f$ is created from base decompositions $\vec{x}',\vec{y}'$ for strategy profiles $\vec{x}$ and $\vec{y}$.
	\item It satisfies capacities $\hat{c}_{\vec{x},\vec{y}}(e,e')$ of $D(\vec{x},\vec{y})$, as $\vec{x}+ \sum_{(B,B') \in  \mathcal{B}^2_{e,e'}} t_{B,B'} \cdot \left(\chi_{e'} - \chi_{e} \right)$
	is a convex combination of bases, and thus an element of $P_{d \cdot rk}$. Therefore, $$f_{(e,e')} = \sum_{(B,B') \in \mathcal{B}^2_{e,e'} } t_{B,B'} <  \hat{c}_{\vec{x},\vec{y}}(e,e').$$
\end{enumerate}
Hence, $f$ is a feasible flow in $D(\vec{x},\vec{y})$, satisfying all supplies and demands. As $\vec x,\vec y \in P_{d \cdot rk}$ were chosen arbitrarily, $P_{d \cdot rk}$ is a bidirectional flow polymatroid. 
}

\end{proof}
An application of these results can be found in the \emph{spanning tree games}.

\begin{example}[Spanning Tree Games]
Consider a finite set of players $N=\{1, \dots n\}$ and an undirected graph $G=(V,E)$ with non-negative, increasing, differentiable, convex and player specific edge costs functions $c_{i,e}$ for all $e \in E$ and $i \in N$. In a spanning tree game, every player $i$ is associated with a weight $d_i$ and a subgraph $G_i$ of $G$. A strategy for player $i$ is to divide it's weight along the spanning trees of $G_i$, to minimize his total costs. 
{If $G$ is a generalized series parallel graph}, then $P_{d_i \cdot rk_i}$ is a bidirectional flow polymatroid, where $rk_i$ be the rank function for the graphic matroid on subgraph $G_i$, (cf. Figure~\ref{inclusionmatroids}). Theorem~\ref{baseorderableproof} implies that equilibria will be unique.
\end{example}

For graphic matroids, the generalized series-parallel graph is the maximal graph structure that allows for a bidirectional flow between every pair of strategies. 

\begin{theorem}[Korneyenko~\cite{Korneyenko1994}, Nishizeki~\cite{Nishizeki1988}]
A graph is generalized series-parallel if and only if it does not contain the $K_4$ as a minor.
\end{theorem}

Let $rk$ be the rank function for the graphic matroid on the $K_4$, we show that there exist two conflicting strategies $\vec{x},\vec{y} \in P_{rk}$, thus there does not exist a flow $f$ in $D(\vec{x},\vec{y})$.

\begin{example} Polymatroid $(E,rk)$ based on the rank function of the graphic matroid on the $K_4$ is not a bidirectional flow polymatroid. Let the resources be numbered as in Figure~\ref{k4notperfect} and look at the strategies $\vec{x}=(1,1,0,0,0,1)$ and $\vec{y}=(0,0,1,1,1,0)$. Graph $D(\vec{x},\vec{y})$ is depicted in Figure~\ref{k4notperfect}. Then there is no flow $f$ in  $D(\vec{x},\vec{y})$ that satisfies all supplies and demands. Resource 1 and 6 have both a supply of 1 and can only exchange load with resource 4 , which only has demand 1. Thus such a flow $f$ does not exist, and $(E,rk)$ is not a bidirectional flow polymatroid.

\begin{figure}[h]
\centering
\begin{tikzpicture}[main_node/.style={circle,fill=black!10,draw,minimum size=1.5em,inner sep=1pt}, scale=0.7]
\tikzset{edge/.style ={thick,red!100}}
\tikzset{edged/.style ={->,dashed}}
\tikzset{invis/.style ={fill=white!10}}

 \node[main_node] (1) at (-0.3, 0)  {};
 \node[main_node] (2) at (4.3, 0)  {};
 \node[main_node] (3) at (2, 4)  {};
 \node[main_node] (4) at (2, 1.6)  {};

 \node[main_node] (5) at (6, 4)  {1};
 \node[main_node] (6) at (6, 2)  {2};
 \node[main_node] (7) at (6, 0)  {6};

 \node[main_node] (8) at (8, 4)  {3};
 \node[main_node] (9) at (8, 2)  {4};
 \node[main_node] (10) at (8, 0)  {5};

\path [ultra thick](1) edge [] node [below] {1} (2);
\path  [ultra thick](2) edge [] node [above right] {2} (3);
\path  [dashed](3) edge [] node [above left] {3} (1);
\path [dashed](1) edge [] node [below right] {4} (4);
\path [dashed](2) edge [] node [below left] {5} (4);
\path  [ultra thick](3) edge [] node [below left] {6} (4);

\path  (6) edge [] node [below left] {} (8);
\path  (6) edge [] node [below left] {} (10);
\path  (6) edge [] node [below left] {} (9);
\path  (5) edge [] node [below left] {} (9);
\path  (7) edge [] node [below left] {} (9);

\end{tikzpicture}
\caption{Left: the $K_4$ with two strategies $\vec{x}$ (thick), $\vec{y}$ (dashed). Right: $D(\vec{x},\vec{y})$.}
\label{k4notperfect}
\end{figure}
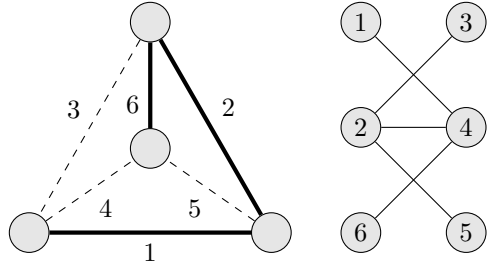
\end{example}

\section{{Non-differentiable Convex Functions}}\label{sec:nondifferentiable}
So far we assumed cost functions to be differentiable. When this is not the case, the proof we gave in the previous section won't hold. When a function is not differentiable, one speaks about the left derivative $c^-(x)$ and the right derivative $c^+(x)$. In the same way we define $\mu_{i,e}^-(\vec{x})=c(x_e) + x_{i,e}c^-(x_e)$ and $\mu_{i,e}^+(\vec{x})=c(x_e) + x_{i,e}c^+(x_e)$. For non-differentiable functions, equilibrium condition \ref{equalresource} generalizes as follows: 

\begin{lemma}[Theorem 8.1, \cite{Fujishige05}]
	\label{equalresourcenondiff}
	Let  $\vec{x}$ be a Nash equilibrium in a polymatroid congestion game.  If $x_{i,e}>0$, then for all $e' \in E$ for which  there is an $\epsilon>0$ such that
	$\vec{x_i}+\epsilon(\chi_{e'}-\chi_e) \in P_{\rho_i}$, we have $\mu^-_{i,e}(\vec{x}) \leq \mu^+_{i,e'}(\vec{x})$.
\end{lemma}

The uniqueness proof in the previous section will not hold as equations~\eqref{equation3} and~\eqref{equation4} might fail to hold using this new equilibrium condition. An example with multiple equilibria is as follows.

\begin{example}
	We look at a two player asymmetric game on three resources. Both players have equal weight 1, and the first player can only use resource 1 or 2, the second player can only use resource 2 or 3. Note that this game is a 1-uniform matroid congestion, and therefore a bidirectional polymatroid. We use the following non-player specific cost functions:
	$$c_1(x_1) = {4}x_1 \; \; \; \; \; c_2(x_2) = \begin{cases} x_2 & \text{ if } x_2<1 \\ 10x_2- 9 & \text{ otherwise }   \end{cases} \; \; \; \; \; c_3(x_3) = {4}x_3. $$
	Let $(x_{1,1},x_{1,2},x_{2,2},x_{2,3})$ denote a game. Then both $(\frac{1}{3},\frac{2}{3},\frac{1}{3},\frac{2}{3})$ and $(\frac{2}{3},\frac{1}{3},\frac{2}{3},\frac{1}{3})$ correspond to Nash equilibria. Note that these are two Nash equilibria where the total load on the resources differ.
\end{example}

The same example can be modified for one with symmetric strategy spaces, but player specific costs, by incurring a high cost on the unavailable resources. The question remains unresolved for symmetric player specific cost functions.

\section{Non-Matroid Set Systems}\label{sec:nonmatroids}
We now derive necessary conditions on a given set system $(\S_i)_{i\in N}$ 
so that any atomic splittable congestion game based on $(\S_i)_{i\in N}$ 
admits unique equilibria.
We show that the \emph{matroid property} is a necessary condition on the players' strategy spaces that guarantees uniqueness of equilibria \emph{without} taking into account how the strategy spaces of different players interweave.\footnote{The term ``interweaving" has been introduced by Ackermann et al.~\cite{Ackermann08,Ackermann09}.}
To state this property mathematically precisely, 
we introduce the notion of \emph{embeddings} {of $\S_i$ in  $E$.
An embedding is a map $\tau:=(\tau_i)_{i\in N}$, where every $\tau_i: E_i \rightarrow E$ is an injective map from $E_i:=\cup_{S\in\S_i} S$ to $E$.
For $X\subseteq E_i$, we denote $\tau_i(X):=\{\tau_i(e), e\in X\}$. Mapping $\tau_i$ induces an isomorphism $\phi_{\tau_i}:\S_i \rightarrow \S'_i$ with $ S\mapsto\tau_i(S)$ and $\S'_i:=\{\tau_i(S)|S\in \mathcal{\S}_i\}$. Isomorphism $\phi_{\tau}=(\phi_{\tau_i})_{i\in N}$ induces the isomorphic strategy space $\phi_{\tau}(\S)=(\phi_{\tau_i}(\S_i))_{i\in N}$. }

\begin{definition}\label{def:embeddings}
A family of set systems $\S_i\subseteq 2^{E_i}$, for $ i\in N$ is said to have the \emph{strong uniqueness property} if for all embeddings $\tau$, the induced game with isomorphic strategy space $\phi_{\tau}(\S)$ has unique Nash equilibria.
\end{definition}

Since for bases of matroids any embedding $\tau_i$ {with isomorphism $\phi_{\tau_i}$} 
has the property that $\phi_{\tau_i}(\S_i)$ is again a collection of bases of a matroid, we obtain the
following immediate consequence of Theorem~\ref{uniquenessresult}.
\begin{corollary}\label{cor:strong-uniqueness}
If $(\S_i)_{i\in N}$ consists of bases
of a base-orderable matroid $M_i=(E,\I_i),$ $ i\in N$, then $(\S_i)_{i\in N}$
possess the strong uniqueness property.
\end{corollary}

For obtaining necessary conditions we need a certain property of non-matroids stated in the following {lemma}. 

\begin{lemma}{{\cite[Lemma 16]{Ackermann09}}}\label{l.anti-matroid}
 If $\S_i\subseteq 2^{E_i}$  with $\S_i\neq \emptyset$ is a non-matroid, then there exist $X,Y\in \S_i$ and $\{a,b,c\}\subseteq X\Delta Y:= (X\setminus Y) \cup (Y\setminus X)$
such that for each set $Z\in \S_i$ with $Z\subseteq X\cup Y$, either
$a\in Z$ or $\{b,c\} \subseteq Z$.
\end{lemma}

\begin{theorem}\label{thm:nonmatroid}
Let $|N|\geq 3$ and assume that for all $i\in N$, $\S_i$ is 
a non-matroid set system. Then, $(\S_i)_{i\in N}$
does not have the strong uniqueness property.
\end{theorem}
\begin{proof}
We will show that there are embeddings $\tau_i:E_i\to E$, $i\in N$, such that the isomorphic game
$\phi_{\tau}(\S)=(\phi_{\tau_1}(\S_1),\ldots, \phi_{\tau_n}(\S_n))$ admits multiple equilibria.

We can assume w.l.o.g.\ that each set system $\S_i$ forms an anti-chain (in the sense that $X\in \S_i, X\subset Y$ implies $Y\not\in \S_i$) since cost functions are non-negative and strictly increasing.
Let us call a non-empty set system $\S_i\subseteq 2^{E_i}$ a \emph{non-matroid} if $\S_i$ is an anti-chain and
$(E_i, \{X\subseteq S : S\in \S_i\})$ is not a
matroid.

Let $\tilde{E}=\bigcup_{i\in N} \tau_i(E_i)$ denote the set of all resources under the embeddings $\tau_i, i\in N$.
The costs on all  resources in $\tilde{E}\setminus{(\tau_1(E_1)\cup \tau_2(E_2)\cup \tau_3(E_3))}$ are  set to zero. Also,
the demands of all players $d_i$ with $i\in N\setminus{\{1,2,3\}}$ are set to zero.
This way, the game is basically determined by 
the players $1,2,3$.
We set the demands $d_1=d_2=d_3=1$.

Let us choose two sets $X, Y$ in $\S_1$ and
$\{a, b, c\}\subseteq X\cup Y$
as described in  Lemma~\ref{l.anti-matroid}.
Let $e:=\tau_1(a), f:=\tau_1(b)$ and $g:=\tau_1(c)$.
We set the costs of all resources in $\tau_1(E_1)\setminus{(\tau_1(X)\cup \tau_1(Y))}$ to some very large cost $M$ (large enough so that player $1$ would never use any of these resources).
The cost on all resources in $(\tau_1(X)\cup \tau_1(Y))\setminus{\{e,f,g\}}$ is set to zero.
This way, player $1$ always chooses a strategy $\tau_1(Z)\subseteq \tau_1(X)\cup \tau_1(Y)$ which, by Lemma~\ref{l.anti-matroid}, either contains
$e$, or it contains both $f$ and $g$. We apply the same construction for player $2$ and $3$,
only changing the role of $e$ to act as $f$ and $g$, respectively.

Note that the so-constructed game is essentially isomorphic to the routing game illustrated in Figure \ref{counterexample2}
if we interpret resource $e$ as arc $(s_1,t_1)$, resource $f$ as arc $(s_2,t_2)$, and resource $g$ as arc $(s_3,t_3)$.
On every edge there is a player specific cost function, given in Table~\ref{tablecounterexample}.

\begin{table}
\centering
\caption{Cost functions used for constructing a game with multiple equilibria.}
\label{tablecounterexample}
\begin{tabular}{c | c c c}
                	& e 				& f 				& g\\ \hline
Player 1 	&$c_{1,e}(x)=x^3$	& $c_{1,f}(x)=x+1$ 	& $c_{1,g}(x)=x+1$ \\ 
Player 2 	&$c_{2,e}(x)=x+1$	& $c_{2,f}(x)=x^3$ 	& $c_{2,g}(x)=x+1$ \\   
Player 3 	&$c_{3,e}(x)=x+1$ 	& $c_{3,f}(x)=x+1$ 	& $c_{3,g}(x)=x^3$ \\ 
\end{tabular}
  
\end{table}

Every player has two possible paths: the direct path that uses only one edge, or the indirect path that uses two edges. We show that the game where everyone puts all their weight on the direct path is a Nash equilibrium, as is the game where everybody puts their weight on the indirect path.

\begin{figure}[h]
\centering
\begin{tikzpicture}[main_node/.style={circle,fill=black!10,draw,minimum size=1em,inner sep=1pt}]
\tikzset{edge/.style ={}}
\tikzset{cut/.style ={thick,fill=red!100,bend right}}
\tikzset{invis/.style ={fill=white!10}}
    \node[main_node] (1) at (0, 0) {$s_1,t_3$};
    \node[main_node] (2) at (2, 2) {$s_2,t_1$};
    \node[main_node] (3) at (-1, 3) {$s_3,t_2$};

\path[] (1) edge [] node [below right] {$e$} (2);
\path[] (2) edge [] node [above] {$f$} (3);
\path[] (3) edge [] node [below left] {$g$} (1);

\end{tikzpicture}
\caption{{Counterexample}}
\label{counterexample2}
\end{figure}
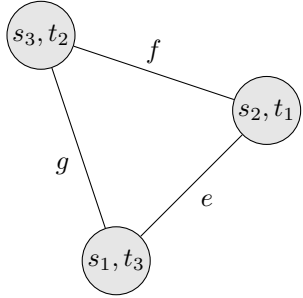
If all players put their weight on the direct route, then player $1$ cannot deviate to decrease it's costs, as:
$$c_{1,e}(1)+c'_{1,e}(1) \cdot 1 = 1+3 \leq 2 + 2 = c_{1,f}(1) + c_{1,g}(1).$$
On the other hand, when all players put their weight on the indirect direct route, player $1$ can also not deviate, as:
$$c_{1,f}(2)+c'_{1,f}(2) \cdot 1 + c_{1,g}(2)+c'_{1,g}(2) \cdot 1= 3+1+3+1 \leq 8 = c_{1,e}(2).$$
The same inequalities hold for player 2 and 3. And therefore everyone playing the direct route, or everyone playing the indirect route both results in a Nash equilibrium.

\end{proof}

\section{A Characterization for  Undirected Graphs}\label{sec:undirected}
In Section~\ref{sec:nonmatroids} we proved that non-matroid set systems in general do not have the strong uniqueness property when there are at least three players, by constructing embeddings $\tau_i$ that lead to the counterexample in Figure~\ref{counterexample2}. This example also gives new insights in uniqueness of equilibria in network congestion games. In the following, we give a characterization of graphs that guarantee uniqueness of Nash equilibria even when player-specific cost functions are allowed. 
\begin{definition}
An undirected graph $G$ is said to have the \emph{uniqueness property} if for any atomic splittable network congestion game on $G=(V,E)$, equilibria are unique.
\end{definition}
Note that in the above definition, we do not specify how source- and sink vertices
are distributed in $V$. We obtain the following result which is related to Theorem 3 of Meunier and Pradeau~\cite{Meunier12}, where a similar result is given for non-atomic congestion games with 
player-specific cost functions.
\begin{theorem}
An undirected graph has the uniqueness property if and only if $G$ has no cycle of length 3 or more.
\end{theorem}
\begin{proof}
Let $G=(V,E)$ be the network in an atomic splittable network congestion game. Assume there exists a cycle $C$ in $G$ of length $k$, with $k\geq 3$. Already for three players, we can create a game with multiple equilibria by generalizing the previous example visualized in Figure~\ref{counterexample2}. Pick three clockwise adjacent vertices $v_1,v_2,v_3$ in cycle $C$ and create three players which have equal weight 1. Player 1 has source $v_1$ and sink $v_2$, player 2 has source $v_2$ and sink $v_3$ and player 3 has source $v_3$ and sink $v_1$. Let $c_{i,e}(x)$ be the cost function for player $i$ at resource $e$. Define $c_{i,e}(x)$ as in Table~\ref{tab:costs}.

\begin{table}
\centering
\caption{Cost functions for a game with multiple equilibria, $M$ is sufficiently large.}\label{tab:costs}
\begin{tabular}{c | c | c | c | c}
 $c_{i,e}(x)$               	& $(v_1,v_2)$		& $(v_2,v_3)$ 		& $C \setminus \{(v_1,v_2),(v_2,v_3)\}$  & $e \notin C$ \\ \hline
Player 1 			&$x^3$	& $x+1$ 	& $\frac{1}{k-2} (x+1)$ & $x+M$ \\ 
Player 2 			&$x+1$	& $x^3$ 	& $\frac{1}{k-2} (x+1)$ & $x+M$ \\   
Player 3 			&$x+1$ 	& $x+1$ 	& $\frac{1}{k-2} \; x^3$ & $x+M$ \\ 
\end{tabular}
\end{table}
For the same reason as in Example \ref{counterexample2} this game has two Nash equilibria: one where all players send their flow clockwise, another where all players send all flow {counterclockwise}.

On the other hand, assume no cycle of length 3 or more in $G$ exists, then $G$ is a tree with parallel edges. Thus, for every source $s$ and sink $t$, there is a unique path from $s$ to $t$ in $G$ modulo parallel edges. Therefore, players only have to decide on how to divide their weight over every set of parallel edges they encounter. As the total cost for a player is just the sum of the costs for all resources separately, players compete only in sets of parallel edges. Atomic splittable congestion games on parallel edges with player-specific cost functions are proven to have a unique Nash equilibrium by Orda et al.~\cite{Orda93}. Thus when $G$ does not contain cycles of length 3 or more, Nash equilibria are unique.
\end{proof}

\section*{Acknowledgements}
We thank Umang Bhaskar and Britta Peis for fruitful discussions. We also thank Neil Olver for pointing out the connection to base orderable matroids.

\bibliographystyle{abbrv}
\bibliography{masterbib}

\appendix
\section{Subclasses of base orderable matroids}
\label{inclusionsproof}
We give proofs {of} the inclusions given in Figure~\ref{inclusionmatroids}:
$$\begin{tabularx}{\textwidth}{ l X}
 $\subset_1:$ & A uniform matroid is a partition matroid where the partition contains only one set. \\ 
 $\subset_2:$ & A partition matroid is a laminar matroid if all sets in the laminar family are disjoint. \\
 $\subset_3:$ & A partition matroid is a  transversal matroid where the sets that need to be traversed are either equal {or} disjoint. \\
 $\subset_4:$ & For the laminar matroid, let $\mathcal{F}$ be the underlying laminar family on ground set $S$ with $S \in \mathcal{F}$. Copy each set $X$ in $\mathcal{F}$ exactly $k_X$ times to create multi set $\mathcal{F}'$, where $k_X$ is the number of elements we are allowed to take from set $X$. Now create a directed graph $G=(V,A)$, where $V=\mathcal{F}' \cup S$, and $A=\{(X,Y) \in \mathcal{F}'  \times \mathcal{F}' | X \subseteq Y \} \cup \{ (s,X)\in S \times \mathcal{F}' | s\in X \}$. Let $U$ be the maximal multi set containing only $S$. Then clearly $G$ with starting points $S$ and endpoints $U$ form a gammoid that corresponds to the laminar matroid. \\
 $\subset_5:$ & A transversal matroid is a gammoid according to Corollary 39.5a in  \cite{schrijver2003combinatorial}.\\
 $\subset_6:$ & Every binary matroid is a gammoid if and only if it is a graphic matroid on a generalized series-parallel graph~\cite{Welsh2010}. As every graphic matroid is binary~\cite{Harary1969}, the graphic matroid on a generalized series-parallel graph is a binary gammoid, and thus a gammoid.\\
 $\subset_7:$ & A gammoid is strongly base orderable according to Theorem 42.12 in \cite{schrijver2003combinatorial}.\\
 $\subset_8:$ & A matroid $\mathcal{M}=(R,\mathcal{I})$ is called \emph{strongly base orderable} (SBO) if for every pair of bases $(B,B')$ there exists a bijective function $g:B \rightarrow B'$ such that $B-X+g(X) \in \mathcal{B}$. Take $X={e}$ and $X=B\setminus \{e\}$ to obtain the conditions for base orderable matroids.\\
\end{tabularx}$$

\end{document}